\pdfoutput=1
\documentclass[twocolumn,amsmath,amssymb,aps,secnumarabic,%
    nofootinbib,groupedaddress,letterpaper]{revtex4-1}
\usepackage{times,mathptmx,amsthm}
\newtheorem{theorem}{Theorem}[section]

\newtheorem{proposition}[theorem]{Proposition}
\newtheorem{corollary}[theorem]{Corollary}
\theoremstyle{definition}

\newtheorem{algorithm}[theorem]{Algorithm}

\newcommand{\C}{\mathbb{C}}
\newcommand{\Z}{\mathbb{Z}}
\renewcommand{\tensor}{\otimes}
\renewcommand{\hat}{\widehat}
\renewcommand{\tilde}{\widetilde}

\newcommand{\cE}{\mathcal{E}}
\newcommand{\cH}{\mathcal{H}}

\newcommand{\ket}[1]{|#1\rangle}
\newcommand{\bra}[1]{\langle#1|}
\newcommand{\braket}[1]{\langle#1\rangle}
\newcommand{\tO}{\tilde{O}}

\renewcommand{\bar}{\overline}
\newcommand{\hb}{\hat{b}}
\newcommand{\hA}{\hat{A}}
\newcommand{\unif}{{\mathrm{unif}}}
\newcommand{\poly}{{\mathrm{poly}}}
\newcommand{\new}{{\mathrm{new}}}
\newcommand{\ceil}[1]{\lceil {#1} \rceil}
\newcommand{\vj}{\vec{j}}

\newcommand{\ie}{\textit{i.e.}}

\renewcommand{\sec}[1]{Section~\ref{#1}}

\newcommand{\cor}[1]{Corollary~\ref{#1}}
\newcommand{\prop}[1]{Proposition~\ref{#1}}
\newcommand{\alg}[1]{Algorithm~\ref{#1}}
\newcommand{\eq}[2]{\begin{equation}\label{#1}#2\end{equation}}
\newcommand{\defeq}[2]{\stackrel{\mathrm{def}}{=}}

\begin{document}
\title{Another subexponential-time quantum algorithm for the 
    dihedral hidden subgroup problem}

\author{Greg Kuperberg}
\email{greg@math.ucdavis.edu}
\thanks{Partly supported by NSF grant DMS CCF-1013079}
\affiliation{Department of Mathematics, University of
    California, Davis, CA 95616}

\begin{abstract} We give an algorithm for the hidden subgroup problem for
the dihedral group $D_N$, or equivalently the cyclic hidden shift problem,
that supersedes our first algorithm and is suggested by Regev's algorithm.
It runs in $\exp(O(\sqrt{\log N}))$ quantum time and uses $\exp(O(\sqrt{\log
N}))$ classical space, but only $O(\log N)$ quantum space.  The algorithm
also runs faster with quantumly addressable classical space than with fully
classical space.  In the hidden shift form, which is more natural for this
algorithm regardless, it can also make use of multiple hidden shifts.
It can also be extended with two parameters that trade classical space
and classical time for quantum time.  At the extreme space-saving end, the
algorithm becomes Regev's algorithm.  At the other end, if the algorithm is
allowed classical memory with quantum random access, then many trade-offs
between classical and quantum time are possible.
\end{abstract}
\maketitle

\section{Introduction}
\label{s:intro}

In a previous article \cite{Kuperberg:dhsp}, we established a
subexponential-time algorithm for the dihedral hidden subgroup problem,
which is equivalent to the abelian hidden shift problem.  That algorithm
requires $\exp(O(\sqrt{\log\ N}))$ time, queries, and quantum space to
find the hidden shift $s$ in the equation $g(x) = f(x+s)$, where $f$ and
$g$ are two injective functions on $\Z/N$.  In this article we present
an improved algorithm, \alg{a:outer}, which is much less expensive
in space, as well as faster in a heuristic model.  Our algorithm was
inspired by and generalizes Regev's algorithm \cite{Regev:dhsp}.  It uses
$\exp(O(\sqrt{\log\ N}))$ \emph{classical} space, but only $O(\log N)$
\emph{quantum} space.  We heuristically estimate a total computation time
of $\tO(2^{\sqrt{2\log_2 N}})$ for the new algorithm; the old algorithm
takes time $\tO(3^{\sqrt{2\log_3 N}})$.

The algorithm also has two principal adjustable parameters.  One parameter
allows the algorithm to use less space and more quantum time.  A second
parameter allows the algorithm to use more classical space and classical
time and less quantum time, if the classical space has quantum access
\cite{GLM:arch}.  (See also \sec{s:timespace}.)  Finally, the new algorithm
can take some advantage of multiple hidden shifts; somewhat anomalously,
our old algorithm could not.

The new algorithm can be called a \emph{collimation sieve}.  As in the
original algorithm and Regev's algorithm, the weak Fourier measurement
applied to a quantum query of the hiding function yields a qubit whose
phases depend on the hidden shift $s$.  The sieve makes larger qudits from
the qubits which we call \emph{phase vectors}.  It then collimates the
phases of the qudits with partial measurements, until a qubit is produced
whose measurement reveals the parity of $s$.  We also use a key idea from
Regev's algorithm to save quantum space.  The sieve is organized as a tree
with $O(\sqrt{\log N})$ stages, and we can traverse the tree depth first
rather than breadth first.  The algorithm still uses a lot of classical
space to describe the coefficients of each phase vector when it lies in a
large qudit.  If the qudit has dimension $\ell$, then this is only $O(\log
\ell)$ quantum space, but the classical description of its phases requires
$\tO(\ell)$ space.

The main discussion of the dihedral hidden subgroup problem has been as
an algorithm with a black-box hiding function.  Recently Childs, Jao, and
Soukharev \cite{CJS:isogenies} found a classical, white-box instance of
the dihedral hidden subgroup problem, or the abelian hidden shift problem.
The instance is that an isogeny between isogenous, ordinary elliptic
curves can be interpreted as a hidden shift on a certain abelian group.
Thus, just as Shor's algorithm allows quantum computers to factor large
numbers, an abelian hidden shift algorithm allows quantum computers to
find isogenies between large elliptic curves.  This is a new impetus to
study algorithms for the dihedral hidden shift problem.

Before describing the algorithm, we review certain points of quantum
complexity theory in general, and quantum algorithms for hidden structure
problems.  We adopt the general convention that if $X$ is a finite set
of orthonormal vectors in a Hilbert space $\cH$ (but not necessarily a
basis), then
$$\ket{X} \defeq \frac{1}{\sqrt{|X|}} \sum_{x \in X} \ket{x}$$
is the constant pure state on $X$.  Also if $X$ is an abstract finite set,
then $\C[X]$ is the Hilbert space in which $X$ is an orthonormal basis.
Also we use the notation
$$[n] = \{0,1,\ldots,n-1\},$$
so that $\C[[n]]$ becomes another way to write the vector space $\C^n$.

\acknowledgments

The author would like to thank Scott Aaronson and Oded Regev for
useful discussions.

\section{Quantum time and space}
\label{s:timespace}

As with classical algorithms, the computation ``time" of a quantum algorithm
can mean more than one thing.  One model of quantum computation is a quantum
circuit that consists of unitary operators and measurements, or even general
quantum operations, and is generated by a classical computer.  (It could
be adaptively generated using quantum measurements.)  Then the circuit
depth is one kind of quantum time, a type of parallel time.  The circuit
gate complexity is another kind of quantum time, a type of serial time.
We can justify serial quantum time with the following equivalence with a
RAM-type machine.

\begin{proposition} The gate complexity of a classically uniform family of
quantum circuits is equivalent, up to a constant factor, to the computation
time of a RAM-type machine with a classical address register, a quantum
data register, a classical tape, and a quantum tape.
\label{p:serial} \end{proposition}

We will discuss \prop{p:serial} more rigorously in \sec{s:rigor}.
From either the circuit viewpoint or the RAM machine viewpoint, serial
computation time is a reasonable cost model:  in practice, gate operations
are more expensive than simple memory multiplied by clock time.

An interesting and potentially important variation of the random-access
model is quantum random access memory, or QRAM \cite{GLM:arch}.  In this
model, there is an address register composed of qubits and a memory can be
accessed in quantum superposition, whether or not the cells of the memory
tape are classical.  Of course, if the memory is classical, only read
operations can be made in quantum superposition.  A RAM quantum computer
thus has four possible types of memory tapes:  classical access
classical memory (CRACM), quantum access classical memory (QRACM),
classical access quantum memory (CRAQM), and quantum access
quantum memory (QRAQM).

Hypothetically, one could cost quantum access classical memory (QRACM) simply
as quantum memory.  But for all we know, quantum access classical memory
(QRACM) and classical-access quantum memory (CRAQM) are non-comparable
resources.  We agree with the suggestion \cite{BHT:collision} that
quantum-access classical memory could be cheaper than quantum memory with
either classical or quantum access.  After all, such memory does not need
to be preserved in quantum superposition.  Our own suggestion for a QRACM
architecture is to express classical data with a 2-dimensional grid of
pixels that rotate the polarization of light.  (A liquid crystal display
has a layer that does exactly that.)  When a photon passes through such
a grid, its polarization qubit reads the pixel grid in superposition.
Such an architecture seems easier to construct than an array of full qubits.

A good example of an algorithm that uses QRACM is the Brassard-H{\o}yer-Tapp
algorithm for the 2-to-1 collision problem \cite{BHT:collision}, as the
authors themselves point out.  Given a function $f:X \to Y$ where $X$ has
$N$ elements, the algorithm generates $N^{1/3}$ values of $f$ at random and
then uses a Grover search over $N^{2/3}$ values to find a collision; thus
the time complexity is $\tO(N^{1/3})$.  This is a large-memory algorithm,
but the bulk of the memory only needs to be quantumly addressable classical
memory.  By contrast, Ambainis' algorithm \cite{Ambainis:distinctness}
for the single collision problem uses true quantum memory.

\begin{proposition} In the RAM model, a quantum access memory with $N$
quantum or classical cells can be simulated with a classical linear access
memory, with the same cells, with $\tO(N)$ time overhead.
\label{p:address} \end{proposition}

\subsection{Some rigor}
\label{s:rigor}

Here we give more precise definitions of quantum RAM machine models, and
we argue Propositions~\ref{p:serial} and \ref{p:address}.  We would like
models that have no extraneous polynomial overhead, although they might have
polylogarithmic overhead.  On the other hand, it seems very difficult to
regularize polylogarithmic overhead.  In our opinion, different models of
computation that differ in polylogarithmic overhead could be equally good.
Actually, at some level a physical computer has at most the computational
strength of a 3-dimensional cellular automaton, where again, the total
number of operations is as important as the total clock time.  (Or even
a 2-dimensional cellular automaton; a modern computer is approximately a
2-dimensional computer chip.)  Procedural programming languages typically
create a RAM machine environment, but usually with polylogarithmic overhead
that depends on various implementation details.

A classical Turing machine $M$ is a tuple $(S,\Gamma,\delta)$, where $S$
is a finite set of states, $\Gamma$ is a finite alphabet, and $\delta$
is a transition map.  The Turing machine has a tape which is linear in
one direction with a sequence of symbols in $\Gamma$, which initially
are all the blank symbol $b \in \Gamma$ except for an input written
in the alphabet $\Sigma = \Gamma \setminus \{b\}$.  The state set $S$
includes an initial state, a ``yes" final state, and a ``no" final state.
Finally the transition map $\delta$ instructs the Turing machine to change
state, write to the tape, and move along the tape by one unit.

In one model of a RAM machine, it is a Turing machine $M$ with two tapes,
an address tape $T_A$ with the same rules as a usual linear tape; and a
main work tape $T_W$.  The machine $M$ (as instructed by $\delta$) can now
also read from or write to $T_W(T_A)$, meaning the cell of the tape $T_W$
at the address expressed in binary (or some other radix) on the tape $T_A$.
It is known \cite{LL:tree,PR:space} that a RAM machine in this form is
polylog equivalent to a \emph{tree Turing machine}, meaning a standard
Turing machine whose tape is an infinite rooted binary tree.

It is useful to consider an intermediate model in which the transition map
$\delta$ is probabilistic, \ie, a stochastic matrix rather than a function.
(Or a substochastic matrix rather than a partial function.)  Then the
machine $M$ arrives at either answer, or fails to halt, with a well-defined
probability.  This is a non-deterministic Turing machine, but it can still
be called classical computation, since it is based on classical probability.

One workable model of a RAM quantum computer is all of the above, except
with two work tapes $T_C$ and $T_Q$, and a register (a single ancillary
cell) $R_Q$.  In this model, each cell of $T_Q$ has the Hilbert
space $\C[\Gamma]$, and the cell $R_Q$ does as well.  The machine $M$
can apply a joint unitary operator (or a TPCP) to the state of $R_Q$ and
the state of the cell of $T_Q$ at the classical address in $T_A$.
Or it can decide its next state in $S$ by measuring the state in $R_Q$.
Or it can do some classical computation using the classical tape $T_C$
to decide what to do next.  All of this can be arranged so that
$\delta$ is a classical stochastic map (which might depend on quantum
measurements), $T_A$ and $T_C$ are classical but randomized, and all of
the quantum nondeterminism is only in the tape $T_Q$ and the register $R_Q$.
In some ways this model is more complicated than necessary, but it
makes it easy to keep separate track of quantum and classical resources.
$T_C$ is a CRACM and $T_Q$ is a CRAQM.

\prop{p:serial} is routine in this more precise model.  The machine can
create a quantum circuit drawn from a uniform family using $T_A$ and $T_C$.
Either afterwards or as it creates the circuit, it can implement it with
unitary operations or quantum operations on $T_Q$ and $R_Q$.  Finally it
can measure $R_Q$ to decide or help decide whether to accept or reject
the input.  At linear time or above, it doesn't matter whether the input
is first written onto $T_C$ or $T_Q$.

The basic definition of quantum addressability is to assume that the address
tape $T_A$ is instead a quantum tape.  For simplicity, we assume some
abelian group structure on the alphabet $\Gamma$.  Then adding the value
of $T_C(T_A)$ to $R_Q$ is a well-defined unitary operator on the joint
Hilbert space of $T_A$ and $R_Q$; in fact it is a permutation operator.
This is our model of QRACM.  Analogously, suppose that we choose a unitary
operator $U_{QR}$ that would act on the joint state of $T_Q(T_A)$ and $R_Q$
if $T_A$ were classical.  Then it yields a unitary operator $U_{QAR}$
on the joint state of $T_Q$, $T_A$, and $R_Q$ that, in superposition,
applies $U_{QR}$ to $T_Q(T_A)$ and $R_Q$.  This is a valid model of QRAQM.

To prove \prop{p:address}, we assume that $T_C$ can no longer be addressed
with $T_A$, and that instead the Turing machine has a position $n$ on
the tape $T_C$ that can be incremented or decremented.  Then to emulate
a quantum read of $T_C(T_A)$, the machine can step through the tape $T_C$
and add $T_C(n)$ to $R_Q$ on the quantum condition that $n$ matches $T_A$.
This is easiest to do if the machine has an auxiliary classical tape that
stores $n$ itself.  Even otherwise, the machine could space the data on
$T_C$ so that it only uses the even cells, and with logarithmic overhead
drag the value of $n$ itself on the odd cells.

\section{Hide and seek}
\label{s:hide}

\subsection{Hidden subgroups}
\label{s:hidegroup}

This section is strictly a review of ideas discussed in our earlier article
\cite{Kuperberg:dhsp}.

In the usual hidden subgroup problem, $G$ is a group, $X$ is an unstructured
set, and $f:G \to X$ is a function that hides a subgroup $H$.  This means
that $f$ factors through the coset space $G/H$ (either left or right cosets),
and the factor $f:G/H \to X$ is injective.  In a quantum algorithm to find
the subgroup $H$, $f$ is implemented by a unitary oracle $U_f$ that adds
the output to an ancilla register.  More precisely, the Hilbert space
of the input register is the group algebra $\C[G]$ when $G$ is finite
(or some finite-dimensional approximation to it when $G$ is infinite),
the output register is $\C[X]$, and the formula for $U_f$ is
$$U_f\ket{g,x_0} = \ket{g,f(g)+x_0}.$$

All known subexponential algorithms for the hidden subgroup problems make
no use of the output when the target set $X$ is unstructured.  (We do not
know whether it is even possible to make good use of the output with only
subexponentially many queries.)  The best description of what happens is
that the algorithm discards the output and leave the input register in
a mixed state $\rho$.  However, it is commonly said that the algorithm
measures the output.  This is a strange description if the algorithm
then makes no use of the measurement; its sole virtue is that it leaves
the quantum state of the input register in a pure state $\ket{\psi}$.
The state $\ket{\psi}$ is randomly chosen from a distribution, which is
the same as saying that the register is in a mixed state $\rho$.

If the output of $f$ is always discarded, then the algorithm works just
as well if the output of $f$ is a state $\ket{\psi(g)}$ in a Hilbert
space $\cH$.  The injectivity condition is replaced by the orthogonality
condition $\braket{\psi(g)|\psi(h)} = 0$ when $g$ and $h$ lie in distinct
cosets of $H$.  In this case $f$ would be implemented by a unitary
$$U_f\ket{g,x_0} = \ket{g} \tensor U_g\ket{x_0},$$
with the condition that if $x_0=0$, then
$$U_g\ket{0} = \ket{\psi(g)}.$$
Or we can have the oracle, rather than the algorithm, discard the output.
In this case, the oracle is a quantum operation (or quantum map) $\cE_{G/H}$
that measures the name of the coset $gH$ of $H$, and only returns the
input conditioned on this measurement.

Suppose that the group $G$ is finite.  Then it is standard to supply
the constant pure state $\ket{G}$ to the oracle $U_f$, and then discard
the output.  The resulting mixed state,
$$\rho_{G/H} = \cE_{G/H}(\ket{G}\bra{G}),$$
is the uniform mixture of $\ket{gH}$ over all (say) left cosets $gH$
of $H$.  This step can also be relegated to the oracle, so that we can
say that the oracle simply broadcasts copies of $\rho_{G/H}$ with no input.

Like our old algorithm, our new algorithm mainly makes use of the state
$\rho_{G/H}$, in the special case of the dihedral group $G = D_N$.
When $N = 2^n$, it is convenient to work by induction on $n$, so that
technically we use the state $\rho_{D_{2^k}/H_k}$ for $1 \le k \le n$.
However, this is not essential.  The algorithm can work in various ways
with identical copies of $\rho_{D_N/H}$.

An important point is that the state $\rho_{G/H}$ is block diagonal with
respect to the weak Fourier measurement on $\C[G]$.  More precisely,
the group algebra $\C[G]$ has a Burnside decomposition
$$\C[G] \cong \bigoplus_V V^* \tensor V,$$
where the direct sum is over irreducible representations of $G$ and also the
direct sum is orthogonal.  The weak Fourier measurement is the measurement
the name of $V$ in this decomposition.  Since $\rho_{G/H}$ is block diagonal,
if we have an efficient algorithm for the quantum Fourier transform on
$\C[G]$, then we might as well measure the name of $V$ and condition the
state $\rho_{G/H}$ to a state on $V^* \tensor V$, because the environment
already knows\footnote{In other words, Schr\"odinger's cat is out of the bag
(or box).} the name of $V$.  Moreover, the state on the ``row space" $V^*$
is known to be independent of the state on $V$ and carry no information
about $H$ \cite{GSVV:quantum}.  So the algorithm is left with the name
of $V$, and the conditional state $\rho_{V/H}$ on $V$.  The difference in
treatment between the value $f(g)$, and the name of the representation $V$,
both of which are classical data that have been revealed to the environment,
is that the name of $V$ is materially useful to existing quantum algorithms
in this situation.  So it is better to say that the name of $V$ is measured
while the value $f(g)$ is discarded.  (In fact, the two measurements don't
commute, so in a sense, they discredit each other.)

\subsection{Hidden shifts}
\label{s:hideshift}

In our earlier work \cite{Kuperberg:dhsp}, we pointed out that if $A$
is an abelian group, then the hidden subgroup problem on the generalized
dihedral group $G =(\Z/2) \ltimes A$ is equivalent to the abelian hidden
shift problem.  The hard case of a hidden subgroup on $G$ consists of
the identity and a hidden reflection.  (By definition, a reflection is
an element in $G \setminus A$, which is necessarily an element of order
2.) In this case, a single hiding function $f$ on $G$ is equivalent to
two injective functions $f$ and $g$ on $A$ that differ by a shift:
$$f(a) = g(a+s).$$
(Note that we allow an algorithm to evaluate them jointly in superposition.)
Finding the hidden shift $s$ is equivalent to finding the hidden reflection.

In this article, we will consider multiple hidden shifts.  By this we mean
that we have a set of endomorphisms
$$\phi_{j \in J}:A \to A$$
and a set of injective functions
$$f_{j \in J}:A \to X$$
such that 
$$f_j(a) = f_0(a+\phi_j(s)).$$
Here $J$ is an abstract finite indexing set with an element $0 \in J$.
We assume that we know each $\phi_j$ explicitly (with $\phi_0 = 0$) and
that we would like to find the hidden shift $s$.  In the cyclic
case $A = \Z/N$, we can write these relations as
$$f_j(a) = f_0(a + r_j s)$$
for some elements $r_j \in \Z/N$.  Note that, for $s$ to be unique, the
maps $\phi_j$ or the factors $r_j$ must satisfy a non-degeneracy condition.
Since we will only address multiple hidden shifts in the initial input
heuristically, we will not say too much about non-degeneracy when $|J| > 2$.
If $|J| = 2$ then $r_1$ or $\phi_1$ must be invertible to make $s$ unique,
in which case we might as well assume that they are the identity.

As a special case, we can look at the hidden subgroup problem in a
semidirect product $G = K \ltimes A$, where $K$ is a finite group, not
necessarily abelian.  Our original algorithm was a sieve that combined
irreducible representations of such a group $G$ to make improved irreducible
representations.  Anomalously, the sieve did not work better when $|K|
> 2$ than in the dihedral case.  The new algorithm can make some use of
multiple hidden shifts, although the acceleration from this is not dramatic.

The principles of \sec{s:hidegroup} apply to the hidden shift or multiple
hidden shift problem.  For the following, assume that $A$ is a finite group.
We write
$$f(j,a) = f_j(a),$$
and we can again make a unitary oracle $U_f$ that evaluates $f$ as follows:
$$U_f\ket{j,a,x_0} = \ket{j,a,f(j,a)+x_0}.$$
Suppose also that we can't make any sense of the value of $f(j,a)$, so we
discard it.  As in \sec{s:hidegroup}, the unitary oracle $U_f$ is thus
converted to a quantum map $\cE$ that makes a hidden measurement of the
value of $f$ and returns only the input registers, \ie, a state in $\C[J]
\tensor \C[A]$.  Suppose that we provide the map $\cE$ with a state of
the form
\eq{e:rho}{\rho = \sigma \tensor (\ket{A}\bra{A})}
where $\sigma$ is some possibly mixed state on $\C[J]$.  As in
\sec{s:hidegroup}, we claim that we might as well measure the Fourier mode
$\hb \in \hA$ of the state $\cE(\rho)$, because the environment already
knows what it is.  To review, the dual abelian group $\hA$ is by definition
the set of group homomorphisms
$$\hb:A \to S^1 \subset \C,$$
and the Fourier dual state $\ket{\hb}$ is defined as
$$\ket{\hb} = \frac{1}{\sqrt{|A|}} \sum_{a \in A} \bar{\hb(a)} \ket{a}.$$
We state the measurement claim more formally.

\begin{proposition} Let $\cE$ be the partial trace of $U_f$ given by
discarding the output, and let the state $\rho$ be as in \eqref{e:rho}.
Then the state $\cE(\rho)$ is block diagonal with respect to the eigenspaces
of the measurement of $\ket{\hb}$.  Also, the measurement has a uniformly
random distribution.
\label{p:}\end{proposition}

\begin{proof} The key point is that $\rho$ is an $A$-invariant state
and $\cE$ is an $A$-invariant map, where $A$ acts by translation on the
$\C[A]$ register.  The state $\ket{A}$ is $A$-invariant by construction,
while $A$ has no action on the $\C[J]$ register.  Meanwhile $\cE$ is
$A$-invariant because it discards the output of $f$, and translation by
$A$ can be reproduced by permuting the values of $f$.  Since $\rho$ is an
$A$-invariant state, and since the elements of $A$ are unitary, this says
exactly that $\rho$ as an operator commutes with $A$.  The eigenspaces of the
action of $A$ on $\C[J] \tensor \C[A]$ are all of the form $\C[J] \tensor
\ket{\hb}$, so the fact that $\rho$ commutes with $A$ is equivalent to the
conclusion that $\rho$ is block diagonal with respect to the eigenspaces
of the measurement $\ket{\hb}$.

To prove the second part, imagine that we also measure $\ket{j}$ on the
register $\C[J]$.  This measurement commutes with both measuring the Fourier
mode $\ket{\hb}$ and measuring or discarding the output register $\C[X]$,
so it changes nothing if we measure $\ket{j}$ first.  So we know $j$,
and since $f_j:A \to X$ is injective, measuring its value is the complete
measurement of $\ket{a}$ starting with the constant pure state $\ket{A}$.
This yields the uniform state $\rho_\unif$ on $\C[A]$, so the value of
$\ket{\hb}$ is also uniformly distributed.
\end{proof}

Suppose further that in making the state $\rho$, the state $\sigma$ on the
$\C[J]$ register is the constant pure state $\ket{J}$.  If the measured
Fourier mode is $\hb \in \hA$, then the state of the $j$ register after
measuring this mode is:
\eq{e:fourier1}{\ket{\psi} \propto \sum_{j \in J} \hb(\phi_j(s))\ket{j}.}
This can be written more explicitly in the cyclic case $A = \Z/N$.  In this
case there is an isomorphism $A \cong \hA$, and we can write any element
$\hb \in \hA$ as
$$\hb(a) = \exp(2\pi i ab/N),$$
and we can also write
$$\phi_j(a) = r_j a$$
for some elements $r_j \in \Z/N$.  So we can then write
\eq{e:fourier2}{\ket{\psi} \propto
    \sum_{j \in J} \exp(2\pi i b r_js)\ket{j}.}
At this point we know both $b$ and each $r_j$, although for different
reasons:  $r_j$ is prespecified by the question, while $b$ was measured
and is uniformly random.  Nonetheless, we may combine these known values
as $b_j = r_jb$ and write:
\eq{e:fourier3}{\ket{\psi} \propto
    \sum_{j \in J} \exp(2\pi i b_js)\ket{j}.}

To conclude, the standard approach of supplying the oracle $U_f$ with
the constant pure state and discarding the output leads us to the state
\eqref{e:fourier1}, or equivalently \eqref{e:fourier2} or \eqref{e:fourier3}.
(Because measuring the Fourier mode does not sacrifice any quantum
information.)  In the rest of this article, we will assume a supply of
states of this type.

\section{The algorithm}
\label{s:algorithm}

\subsection{The initial and final stages}
\label{s:initfin}

For simplicity, we describe the hidden shift algorithm when $A = \Z/N$ and $N
= 2^n$.  The input to the algorithm is a supply of states \eqref{e:fourier3}.
As explained in our previous work \cite{Kuperberg:dhsp}, the problem for any
$A$, even $A$ infinite as long as it is finitely generated, can be reduced
to the cyclic case with overhead $\exp(O(\sqrt{d}))$.  Also for simplicity,
we will just find the parity of the hidden shift $s$.  Also as explained
in our previous work \cite{Kuperberg:dhsp}, if we know the parity of $s$,
then we can reduce to a hidden shift problem on $\Z/2^{n-1}$ and work by
induction.  Finally, just as in our previous algorithm, we seek a wishful
special case of \eqref{e:fourier3}, namely the qubit state
\eq{e:goal}{\ket{\psi} \propto \ket{0} + \exp(2\pi i (2^{n-1})s/2^n)
    \ket{1} = \ket{0} + (-1)^s\ket{1}.}
If we measure whether $\ket{\psi}$ is $\ket{+}$ or $\ket{-}$,
that tells us the parity of $s$.

Actually, although we will give all of the details in base 2, we could
just as well work in any fixed base, or let $N$ be any product of small
numbers.  This generalization seems important for precise optimization for
all values of $N$, which is an issue that we will only address briefly in
the conclusion section.

\subsection{Combining phase vectors}
\label{s:combine}

Like the old algorithm, the new algorithm combines unfavorable qubits states
$\ket{\psi}$ to make more favorable ones in stages, but we change what
happens in each stage.  The old algorithm was called a sieve, because it
created favorable qubits from a large supply of unfavorable qubits, just as
many classical sieve algorithms create favorable objects from a large supply
of candidates \cite{AKS:lattice}.  The new algorithm could also be called
a sieve, but all selection is achieved with quantum measurement instead
of a combination of measurement and matching.  The process can be called
collimation, by analogy with its meaning in optics:  Making rays parallel.

Consider a state of the form \eqref{e:fourier3}, where we write the
coefficient $b_j$ instead as a function $b(j)$, except that we make no
assumption that $b_j = r_jb$ for a constant $b$.  We also assume that the
index set is explicitly the integers from $0$ to $\ell-1$ for some $\ell$,
the \emph{length} of $\ket{\psi}$:
$$J = [\ell] = \{0,1,\ldots,\ell-1\}.$$
We obtain:
$$\ket{\psi} \propto \sum_{0 \le j < \ell} \exp(2\pi i b(j)s/2^n)\ket{j}.$$
Call a vector of this type a \emph{phase vector}.  We view a phase vector
as favorable if every difference $b(j_1)-b(j_2)$ is divisible by many
powers of 2, and we will produce new phase vectors from old ones that
are more favorable.  In other words, we will \emph{collimate} the phases.
The algorithm collimates phase vectors until finally it produces a state
of the form \eqref{e:goal}.  Note that the state $\ket{\psi}$ only changes
by a global phase if we add a constant to the function $b$.  (Or we can
say that as a quantum state, it does not change at all.)  If $2^m|b(j_1)
- b(j_2)$ for some $m \le n$, then we can both subtract a constant from $b$
and divide the numerator and denominator of $b(j)/2^n$ by $2^m$.  So we
can $\ket{\psi}$ as
$$\ket{\psi} \propto \sum_{0 \le j < \ell} \exp(2\pi i b(j)s/2^h)\ket{j},$$
where $h = m-n$ is the \emph{height} of $\ket{\psi}$.  (We do not necessarily
assign the smallest height $h$ to a given $\ket{\psi}$.)  We would like
to collimate phase vectors to produce one with length 2 and height 1
(but not height 0).

Given two phase vectors of height $h$,
\begin{align*}
\ket{\psi_1} &\propto \sum_{0 \le j_1 < \ell_1} \hspace{-.5em} 
    \exp(2\pi i b_1(j_1)s/2^h)\ket{j_1} \\
\ket{\psi_2} &\propto \sum_{0 \le j_2 < \ell_2} \hspace{-.5em} 
    \exp(2\pi i b_2(j_2)s/2^h)\ket{j_2},
\end{align*}
their joint state is a double-indexed phase vector that
also has height $h$:
\begin{align*}
\ket{\psi_1,\psi_2} &= \ket{\psi_1} \tensor \ket{\psi_2} \\
    &\propto \sum_{\substack{0 \le j_1 < \ell_1 \\ 0 \le j_2 < \ell_2}}
    \hspace{-.5em} \exp(2\pi i (b_1(j_1)+b_2(j_2))s/2^h)\ket{j_1,j_2}.
\end{align*}
We can now collimate this phase vector by measuring
$$c \equiv b_1(j_1) + b_2(j_2) \pmod{2^m}$$
for some $m < h$.  Let $P_c$ be the corresponding measurement projection.
The result is another phase vector
$$\ket{\psi} = P_c \ket{\psi_1,\psi_2},$$
but one with a messy indexing set:
$$J = \{(j_1,j_2) | b_1(j_1) + b_2(j_2) \equiv c \pmod{2^m}\}.$$
We can compute the index set $J$, in fact entirely classically, because
we know $c$.  We can compute the phase multiplier function $b$ as the sum
of $b_1$ and $b_2$.  Finally, we would like to reindex $\ket{\psi}$ using
some bijection $\pi:J \to [\ell_\new]$, where $\ell_\new = |J|$.  As we
renumber $J$, we also permute the phase vector $P_c \ket{\psi_1,\psi_2}$.
Then there is a subunitary operator
$$U_\pi:\C^{\ell_1} \tensor \C^{\ell_2} \to \C^{\ell_\new}$$
that annihilates vectors orthogonal to $\C[J]$ and that is unitary on
$\C[J]$.  Then
$$\ket{\psi_\new} = U_\pi \ket{\psi}.$$
The vector $\ket{\psi_\new}$ has height $h-m$.

Actually, collimation generalizes to more than two input vectors.  Given a
list of phase vectors
$$\ket{\psi_1},\ket{\psi_2},\ldots,\ket{\psi_r},$$
and given a collimation parameter $m$, we can produce a collimate state
$\ket{\psi_\new}$ from them.  We summarize the process in algorithm form:
\begin{algorithm}[Collimation] Input: A list of phase vectors
$$\ket{\psi_1},\ket{\psi_2},\ldots,\ket{\psi_r}$$
of length $\ell_1,\ldots,\ell_r$, and a collimation parameter $m$.
\begin{description}
\item[1.] Notionally form the phase vector
$$\ket{\psi} = \ket{\psi_1} \tensor \ket{\psi_2} \tensor \dots \tensor
    \ket{\psi_r}$$
with indexing set
$$[\ell_1] \times [\ell_2] \times \cdots \times [\ell_r]$$
and phase multiplier function
$$b(\vj) = b(j_1,j_2,\ldots,j_r) = b_1(j_1) + b_2(j_2) + \dots + b_r(j_r).$$
\item[2.] Measure $\ket{\psi}$ according to the value of
\eq{e:congvec}{c = b(\vj) \bmod 2^m}
to obtain $P_c\ket{\psi}$.
\item[3.] Find the set $J$ of tuples $\vj$ that satisfy
\eqref{e:congvec}.  Set $\ell_\new = |J|$ and pick a bijection
$$\pi:J \to [\ell_\new].$$
\item[4.] Apply $\pi$ to the value of $b$ on $J$ and apply $U_\pi$ to
$\ket{\psi}$ to make $\ket{\psi_\new}$ and return it.
\end{description}
\label{a:collimate} \end{algorithm}

\alg{a:collimate} is our basic method to collimate phase vectors.
We can heuristically estimate the length $\ell$ by assuming that
$b(\vj)$ is uniformly distributed mod $2^m$.  In this case,
\eq{e:ratio}{\ell_\new \approx 2^{-m}\ell_1\ell_2\dots\ell_r.}
So $\ell$ stays roughly constant when $\ell \approx 2^{m/(r-1)}$.

\subsection{The complexity of collimation}
\label{s:collimate}

\begin{proposition} Let $\ket{\psi_1}$ and $\ket{\psi_2}$ be two phase
vectors of length $\ell_1$ and $\ell_2$ and height $h$, and suppose that
they are collimated mod $2^m$ to produce a phase vector $\ket{\psi_\new}$
of length $\ell_\new$.  Suppose also that the quantum computer is allowed
QRACM.  Then taking $\ell_{\max} = \max(\ell_1,\ell_2,\ell_\new)$ and $r=2$,
\alg{a:collimate} needs
\begin{itemize}
\item $\tO(\ell_{\max})$ classical time (where ``$\tO$" allows factors of
    both $\log \ell_{\max}$ and $h \le n = \log N$).
\item $O(\ell_{\max} h)$ classical space,
\item $O(\ell_{\max} \max(m,\log \ell_{\max}))$ classical space with quantum
    access,
\item $\poly(\log \ell_{\max})$ quantum time, and
\item $O(\log \ell_{\max})$ quantum space.
\end{itemize}
\label{p:collimate} \end{proposition}

\begin{proof} First, we more carefully explain the data structure of a
phase vector $\ket{\psi}$.  The vector $\ket{\psi}$ itself can be stored in
$\ceil{\log_2 \ell_{\max}}$ qubits.  The table $b$ of phase multipliers
is a table of length $O(\ell_{\max})$ whose entries have $h$ bits, so this
is $O(\ell_{\max} h)$ bits of classical space.  \alg{a:collimate} needs the
low $m$ bits of each entry in the table, so $O(\ell_{\max} m)$ bits are kept
in quantum access memory.  We also assume that the table $b$ is sorted on
low bits.

We follow through the steps of \alg{a:collimate}, taking care to manage
resources at each step.  First, measuring
$$c \equiv (b_1(j_1) + b_2(j_2)) \pmod{2^m}$$
can be done in quantum time $\poly(\log \ell,m)$ by looking up the values
and adding them.  As usual, when performing a partial quantum measurement,
the output must be copied to an ancilla and the scratch work (in this case
the specific values of $b_1$ and $b_2$) must be uncomputed.

The other step of collimation is the renumbering.  To review, the measurement
of $c$ identifies a set of double indices
$$J \subseteq [\ell_1] \times [\ell_2].$$
These indices must be renumbered with a bijection
$$\pi:J \to [\ell_\new],$$
indeed the specific bijection that sorts the new phase multiplier table $b =
b_1 + b_2$.  The function $\pi$ can be computed in classical time $\tO(\ell)$
using standard algorithms, using the fact that $b_1$ and $b_2$ are already
sorted.  More explicitly, we make an outer loop over decompositions
$$c = c_1 + c_2 \in \Z/2^m.$$
In an inner loop, we write all solutions to the equations
$$b_1(j_1) \equiv c_1 \pmod{2^m} \qquad b_2(j_2) \equiv c_2 \pmod{2^m}$$
using sorted lookup.  This creates a list of elements of $J$ in some order.
We can write the values of
$$b(j_1,j_2) = b_1(j_1) + b_2(j_3)$$
along with the pairs $(j_1,j_2) \in J$ themselves.  Then $b$ can be
sorted and $J$ can be sorted along with it.

This creates a stored form of the \emph{inverse} bijection $\pi^{-1}$,
which is an ordinary 1-dimensional array.  We will want this, and we will
also want quantum access to the \emph{forward} bijection $\pi$ stored
as an associative array.  Since we will need quantum access to $\pi$,
we would like to limit the total use of this expensive type of space.
We can make a special associative array to make sure that the total extra
space is $O(\ell_{\max} (\log \ell_{\max}))$ bits.  For instance, we can make
a list of elements of $J$ sorted by $(j_1,j_2)$, a table of $\pi$ sorted
in the same order, and an index of pointers from $[\ell_1]$ to the first
element of $J$ with any given value of $j_1$.

The final and most delicate step is to apply the bijection $\pi$ to
$\ket{\psi}$ in quantum polynomial time in $\log \ell$.  Imagine more
abstractly that $\ket{\psi}$ is a state in a Hilbert space $\C^s$ supported
on a subset $X \subseteq [s]$, and that we would like to transform it to
a state in a Hilbert space $\C^t$ supported on a subset $Y \subset [t]$ of
the same size, using a bijection $\pi:X \to Y$.  We use the group structures
$[s] = \Z/s$ and $[t] = \Z/t$, and we assume quantum access to both $\pi$
and $\pi^{-1}$.  Then we will use these two permutation operators acting
jointly on a $\C^s$ register and a $\C^t$ register:
$$U_1\ket{x,y} = \ket{x,y+\pi(x)}
    \qquad U_2\ket{x,y} = \ket{x-\pi^{-1}(y),y}.$$
A priori, $\pi(x)$ is only defined for $x \in X$ and $\pi^{-1}(y)$ is only
defined for $y \in Y$; we extend them by 0 (or extend them arbitrarily)
to other values of $x$ and $y$.  Then clearly
$$U_2 U_1\ket{x,0} = \ket{0,\pi(x)}.$$
Thus
$$\ket{\psi_\new} = U_2 U_1 \ket{\phi,0}$$
is what we want. Following the rule of resetting the height to $0$, we
can also let
$$b_\new(j) = b(j)/2^m.$$
\end{proof}

\begin{corollary} Taking the hypotheses of \prop{p:collimate}, if the
quantum computer has no quantum access memory, then \alg{a:collimate}
can be executed with $r=2$ with
\begin{itemize}
\item $\tO(\ell_{\max})$ quantum time (and classical time),
\item $\tO(\ell_{\max})$ classical space, and
\item $O(\log \ell_{\max})$ quantum space.
\end{itemize}
\label{c:collimate} \end{corollary}

\cor{c:collimate} follows immediately from \prop{p:collimate}
and \prop{p:address}.  The point is that, even though there is a
performance penalty in the absence of quantum access memory, the same
algorithm still seems competitive.

\subsection{The outer algorithm}
\label{s:outer}

In this section we combine the ideas of Sections \ref{s:hideshift},
\ref{s:initfin}, \ref{s:combine}, and \ref{s:collimate} to make a complete
algorithm.  We present the algorithm with several free parameters.
We will heuristically analyze these parameters in \sec{s:heuristic}.
Then in \sec{s:rigor} we will simply make convenient choices for the
parameter to prove that the algorithm has quantum time and classical space
complexity $\exp(O(\sqrt{n}))$.

The algorithm has a recursive subroutine to produce a phase vector of
height 1.  The subroutine uses a collimation parameter $0 < m(h) \le n-h$
and a starting minimum length $\ell_0$.

\begin{algorithm}[Collimation sieve]
Input: A height $h$, a collimation parameter $m = m(h)$, a branching
parameter $r = r(h)$, a starting minimum length $\ell_0$, and access to
the oracle $U_f$.  Goal: To produce a phase vector of height $h$.
\begin{description}
\item[1.] If $h=n$, extract phase vectors
$$\ket{\psi_1},\ket{\psi_2},\ldots,\ket{\psi_s}$$
of height $n$ from the oracle as described in \sec{s:hide} until the
length of
$$\ket{\psi_\new} = \ket{\psi_1,\psi_2,\ldots,\psi_s}$$
is at least $\ell_0$.  Return $\ket{\psi_\new}$.
\item[2.] Otherwise, recursively and \emph{sequentially}
obtain a sequence of phase vectors
$$\ket{\psi_1}, \ket{\psi_2}, \ldots, \ket{\psi_r}$$
of height $h+m$.
\item[4.] Collimate the vectors mod $2^m$ using \alg{a:collimate} to
produce a phase vector $\ket{\psi_\new}$ of height $h$.  Return it.
\end{description}
\label{a:outer} \end{algorithm}

When called with $h=1$, \alg{a:outer} produces a phase vector
$$\ket{\psi} \propto \sum_{0 \le j < \ell} (-1)^{b(j)s}\ket{j}.$$
Otherwise, we pick a maximal subset $X \subseteq [\ell]$ on which $b$
is equally often $0$ and $1$.  (Note that this takes almost no work,
because the collimation step sorts $b$.)  If $X$ is empty, then we must
run \alg{a:outer} again.  Otherwise, we measure whether $\ket{\psi}$ is
in $\C[X]$.  If the measurement fails, then again we must run Subroutine
A again.  Otherwise the measured form of $\ket{\psi}$ has a qubit factor
of the form
$$\ket{0} + (-1)^s\ket{1},$$
and this can be measured to obtain the parity of $s$.

\alg{a:outer} recursively makes a tree of phase vectors that are more and
more collimated, starting with phase vectors obtained from the hiding
function $f(j,a)$ by the weak Fourier measurement.  An essential idea,
which is due to Regev and is used in his algorithm, is that with the
collimation method, the tree can be explored depth-first and does not need
to be stored in its entirety.  Only one path to a leaf needs to be stored.
No matter how the collimation parameter is set, the total quantum space
used is $O(n^2)$, while the total classical space used is $O(n\max(\ell))$.
(But the algorithm is faster with quantum access to the classical space.)

An interesting feature of the algorithm is that its middle part, the
collimation sieve, is entirely \emph{pseudoclassical}.  The algorithm begins
by applying QFTs to oracle calls, as in Shor's algorithm.  It ends with the
same parity measurement as Simon's algorithm.  These parts of the algorithm
are fully quantum in the sense that they use unitary operators that are not
permutation matrices.  However, collimation consists entirely of permutations
of the computational basis and measurements in the computational basis.

\subsection{Heuristic analysis}
\label{s:heuristic}

Heuristically the algorithm is the fastest when $r=2$.

Suppose that the typical running time of the algorithm is $f(n)$, with some
initial choice of $m = m(1)$.  First, creating a phase vector of height $h$
is similar to running the whole algorithm with $n' = n-h$.  So the total
computation time (both classical and quantum) can be estimated as
$$f(n) \approx \min_m \left(2^m + 2f(n-m)\right).$$
Here the first term is dominated by the classical work of collimation,
while the second term is the recursive work.  The two terms of the minimand
are very disparate outside of a narrow range of values of $m$.  So we
can let $g(n) = \log_2 f(n)$, and convert multiplication to addition and
approximate addition by $\max$.  (This type of asymptotic approximation
is lately known in mathematics as \emph{tropicalization}.)  We thus
obtain
$$g(n) \approx \min_m \left(\max(m,g(n-m)+1\right).$$
The solutions to this equation are of the form
$$g(\frac{m(m+1)}2+c) = m,$$
where $c$ is a constant.  We obtain the heuristic estimate
\eq{e:heuristic}{f(n) \stackrel{?}= \tO(2^{\sqrt{2n}})}
for both the quantum plus classical time complexity and the classical space
complexity of the algorithm.  We put a question mark because we have not
proven this estimate.  In particular, our heuristic calculation does not
address random fluctuations in the length estimate \eqref{e:ratio}.

If the quantum computer does not have QRACM or if it is no cheaper than
quantum memory, then the heuristic \eqref{e:heuristic} is the best
that we know how to do.  If the algorithm is implemented with QRACM,
then the purely quantum cost is proportional to the number of queries.
In this case, if there is extra classical space, we can make $m$ larger
and larger to fill the available space and save quantum time.  This is the
``second parameter" mentioned in \sec{s:intro}.  However, this adjustment
only makes sense when classical time is much cheaper than quantum time.
In particular, \eqref{e:heuristic} is our best heuristic if classical and
quantum time are simply counted equally.

If classical space is limited, then equation \eqref{e:ratio} tells us that
we can compensate by increasing $r$.  To save as much space as possible, we
can maintain $\ell = 2$ and adjust in each stage of the sieve $r$ to optimize
the algorithm.  In this case the algorithm reduces to Regev's algorithm.

\subsection{A rigorous complexity bound}
\label{s:bound}

The goal of this section is to rigorously prove $2^{O(\sqrt{n})}$ complexity
bounds for a likely inefficient modified version of \alg{a:outer}.  For
simplicity we assume that $n = m^2$, and we assume two hidden shifts $f_0$
and $f_1$.  In the first stage, we form phase vectors of length $\ell_0 =
2^{m+1}$ from $m+1$ qubits of the form \eqref{e:fourier3}.  We construct
the collimation sieve using \alg{a:collimate} to align $m$ low bits of
the phase multiplier at each stage except the last stage.  Suppose that
the output phase vector $\ket{\psi}$ from a use of \alg{a:collimate}
has length $\ell_1$.  We divide the indexing set of $\ket{\psi}$ into
segments of length $\ell_0$ and a leftover segment of length $\ell_2 <
\ell_0$.  Then we perform a partial measurement corresponding into this
partition into segments.  If the measured segment is the short one, then
the phase vector is simply discarded; in particular if $\ell_1 < \ell_0$,
then the vector is discard with no measurement.  Finally at the last stage,
we measure the phase vector $\ket{\psi}$ according to the value of $m-1$
of the remaining $m$ phase multiplier bits.  After this partial measurement
we have a residual qudit with some $\ell_1$ states.  We pair these states
arbitrarily, leaving one singleton if $\ell_1$ is odd, and again perform
the partial measurement corresponding to this partition.  Assuming that
the residual state is a qubit, and a qubit of the form \eqref{e:goal},
then we use it to measure the parity of the phase shift $s$; otherwise we
discard it and restart the entire computation.

\begin{proposition} The modified form of \alg{a:outer} uses quantum time
and classical space $2^{O(\sqrt{n})}$, and quantum space $O(\log n)$.
\end{proposition}

\begin{proof} At the heuristic level, the bounds are straightforward.
The question is to establish that the modified algorithm succeeds at
each stage with probability bounded away from $0$.  The algorithm can
locally fail in three ways: (1) In the intermediate stages, it can make
a phase vector which is too short, either with $\ell_1 < \ell_0$, or the
trimming measurement could leave the remnant of length $\ell_2 < \ell_0$.
(2) At the final stage, we might have $\ell_1 < 2$, or again have a
remnant after the trimming measurement might have length $\ell_2 < 2$.
(3) If a qubit is produced at the very end, it could have a trivial phase
multiplier rather than the form \eqref{e:goal}.

To address the first problem, we have two phase vectors $\ket{\psi_1}$
and $\ket{\psi_2}$ with tables of phase multipliers $b_1$ and $b_2$.
The combined phase vector $\ket{\psi_1,\psi_2}$ then has the phase
multiplier $b_1+b_2$, and we measure $\ket{\psi_1,\psi_2}$ according to
the low $m$ bits of $b_1+b_2$.  Heuristically we can suppose that $b_1+b_2$
is randomly distributed, even though it cannot be exactly true.  We claim
that at the rigorous level, the modified sieve has an adequate chance of
success regardless of the distribution of $b_1+b_2$.  The phases listed
in $b_1+b_2$ are divided among only $2^m$ buckets.  If we pick an entry at
random, which gives the correct distribution for the partial measurement,
then with probability at least $3/4$, then it lies in a bucket of size
$\ell_1 \ge \ell_0$, regardless of how the entries are distributed.  Then,
if $\ell_1 \ge \ell_0$, the probability that the trimming measurement
creates a phase vector length exactly $\ell_0$ is at least $1/2$.

The second problem is addressed in the same way: The $2^{m+1}$ terms of the
last-stage phase vector are divided among $2^{m-1}$ buckets, so measuring
which bucket produces a qudit of length $\ell_2 \ge 2$ with probability
at least $3/4$.  Then the trimming measurement produces a qubit with
probability at least $2/3$.

Finally the third problem requires some knowledge of the distribution of
the phase multipliers.  The final phase multiplier is either $2^{n-1}$ or
$0$, and the former value is the favorable one that allows us to measure
the parity of $s$.  Recall that the initial phase qubits \eqref{e:fourier3}
had uniformly random phase multipliers; in particular the highest bits are
uniformly random and independent.  All of the decisions in the algorithm
so far depend only on the other bits of the phase multiplier.  The final
phase multiplier is a sum of some of the high bits of the initial phase
multipliers, and is therefore also uniformly random.  So we obtain the
state \eqref{e:goal} with probability $1/2$ at this stage.
\end{proof}

\section{Conclusions}

At first glance, the running time of our new algorithm for DHSP or hidden
shift is ``the same" as our first algorithm, since both algorithms run
in time $2^{O(\sqrt{\log N})}$.  Meanwhile Regev's algorithm runs in time
$2^{O(\sqrt{(\log N)(\log \log N)}}$, which may appear to be almost as fast.
Of course, these expressions hide the real differences in performance between
these algorithms, simply because asymptotic notation has been placed in
the exponent.  All polynomial-time algorithms with input of length $n$
run in time
$$n^{O(1)} = 2^{O(\log n)}.$$
Nonetheless, polynomial accelerations are taken seriously in complexity
theory, whether they are classical or quantum accelerations.

For many settings of the parameters, \alg{a:outer} is superpolynomially
faster than Regev's algorithm.  It is Regev's algorithm if we have
exponentially more quantum time than classical space.  However, in real
life, classical computation time has only scaled polynomially faster than
available classical computer memory.  So it is reasonable to consider a
future regime in which quantum computers exist, but classical memory is
cheaper than quantum time, or is only polynomially more expensive.

Regev \cite{Regev:quantum} established a reduction from certain lattice
problems (promise versions of the short vector and close vector problems)
to the version of DHSP or hidden shift in which $f(a)$ and $g(a+s)$ are
overlapping quantum states.  At first glance, our algorithms apply to this
type of question.  However, we have not found quantum accelerations for these
instances.  The fundamental reason is that we have trouble competing with
classical sieve algorithms for these lattice problems \cite{AKS:lattice}.
The classical sieve algorithms work in position space, while our algorithms
work in Fourier space, but otherwise the algorithms are similar.  Instead,
DHSP seems potentially even more difficult than related lattice problems
(since that is the direction of Regev's reduction) and the main function
of our algorithms is to make DHSP roughly comparable to lattice problems
on a quantum computer.

One significant aspect of \alg{a:outer}, and also in a way Regev's algorithm,
is that it solves the hidden subgroup problem for a group $G = D_N$ without
staying within the representation theory of $G$ in any meaningful way.
It could be interesting to further explore non-representation methods
for other hidden structure problems.


\providecommand{\bysame}{\leavevmode\hbox to3em{\hrulefill}\thinspace}
\providecommand{\MR}{\relax\ifhmode\unskip\space\fi MR }
\providecommand{\MRhref}[2]{%
  \href{http://www.ams.org/mathscinet-getitem?mr=#1}{#2}
}
\providecommand{\href}[2]{#2}
\providecommand{\eprint}{\begingroup \urlstyle{tt}\Url}

\end{document}